\documentclass[conference]{IEEEtran}
\IEEEoverridecommandlockouts
\usepackage{graphicx,cite,calc,xcolor,subfigmat,amssymb,amsmath,mathrsfs,dsfont,epstopdf}
\usepackage[normalem]{ulem}
\usepackage{ulem}
\usepackage{graphicx} % Required for inserting images
\usepackage[switch]{lineno}

\usepackage{cite}
\usepackage[linesnumbered,ruled,vlined]{algorithm2e}

\usepackage{mathtools}
\usepackage{algorithmicx}

\usepackage{dsfont}
\usepackage{bbm}
\usepackage{soul}
\usepackage{xcolor}
\usepackage{amsmath,amsthm}
\usepackage{amsmath}

\usepackage{float}
\usepackage[T1]{fontenc}
\usepackage{soul} 
\usepackage{xcolor}
\usepackage{optidef}
\usepackage{listings}
\usepackage{enumitem}
\newtheorem{theorem}{Theorem}

\newtheorem{prop}{Proposition}

\newtheorem{lemma}{Lemma}
\newtheorem{corollary}{Corollary}[theorem]

\makeatletter
\newcommand*{\rom}[1]{\expandafter\@slowromancap\romannumeral #1@}
\makeatother
\usepackage {amsmath}
\definecolor{highlightcolor}{RGB}{255,255,0}

\usepackage{newtxtext,newtxmath}
\title {Addressing the Load Estimation Problem: Cell Switching in HAPS-Assisted Sustainable 6G Networks}

\begin{document}

 \author{\rm{Maryam Salamatmoghadasi}, \rm{Metin Ozturk}, \rm{Halim Yanikomeroglu}
 \thanks{This research has been sponsored in part by the NSERC Create program entitled TrustCAV and in part by The Scientific and Technological Research Council of Türkiye (TUBITAK).
 
All the authors are with Non-Terrestrial Networks Lab, Department of Systems and Computer Engineering, Carleton University, Ottawa, ON K1S5B6, Canada. Metin Ozturk is also with Electrical and Electronics Engineering, Ankara Yildirim Beyazit University, Ankara, 06010, Turkiye. emails: \texttt{maryamsalamatmoghad@cmail.carleton.ca, metin.ozturk@aybu.edu.tr, halim@sce.carleton.ca}.}}

\maketitle
\begin{abstract}
This study aims to introduce and address the problem of traffic load estimation in the cell switching concept within the evolving landscape of vertical heterogeneous networks (vHetNets).
The problem is that the practice of cell switching faces a significant challenge due to the lack of accurate data on the traffic load of sleeping small base stations (SBSs).
This problem makes the majority of the studies in the literature, particularly those employing load-dependent approaches, impractical due to their basic assumption of perfect knowledge of the traffic loads of sleeping SBSs for the next time slot.
Rather than developing another advanced cell switching algorithm, this study investigates the impacts of estimation errors and explores possible solutions through established methodologies in a novel vHetNet environment that includes the integration of a high altitude platform (HAPS) as a super macro base station (SMBS) into the terrestrial network.
In other words, this study adopts a more foundational perspective, focusing on eliminating a significant obstacle for the application of advanced cell switching algorithms.
To this end, we explore the potential of three distinct spatial interpolation-based estimation schemes: random neighboring selection, distance-based selection, and clustering-based selection.
Utilizing a real dataset for empirical validations, we evaluate the efficacy of our proposed traffic load estimation schemes. Our results demonstrate that the multi-level clustering (MLC) algorithm performs exceptionally well, with an insignificant difference (i.e., 0.8\%) observed between its estimated and actual network power consumption, highlighting its potential to significantly improve energy efficiency in vHetNets.
 
\end{abstract}
\section{Introduction}

With the advent of sixth-generation (6G) cellular networks, expected to support an unprecedented number of devices per square kilometer, the quest for enhanced connectivity is vital. 
However, this expansion faces significant challenges, including a notable increase in energy consumption within radio access networks~(RANs), with base stations~(BSs), particularly small BSs~(SBSs), being major contributors~\cite{3333333}.
Addressing these challenges requires re-evaluating current operational practices and implementing innovative solutions. 
In this regard, strategically deactivating BSs or putting them into sleep mode during low activity periods, i.e., cell/BS switching off~\cite{3333333}, emerges as a feasible solution for enhancing energy efficiency and network sustainability. 
Nonetheless, implementing effective cell switching strategies is impeded by the absence of precise traffic load data for sleeping SBSs at the next time slot, rendering existing energy optimization methods ineffective, since the majority of the works in the literature assume perfect knowledge on the traffic loads of sleeping BSs~\cite{10304250,HETS2023P,ELAA2022JR,EOMK2017JR,Metin_VFA_CellSwitch}.
Overcoming this hurdle by accurately estimating the traffic loads of sleeping SBSs is crucial, bridging the gap between theory and practice and unlocking the full potential of cell switching for significantly improving sustainability in vertical heterogeneous networks (vHetNets), which integrates satellites, high altitude platform stations (HAPS), unmanned aerial vehicles (UAVs), and terrestrial BSs to ensure comprehensive global coverage.
\textit{Therefore, in this study, instead of taking a small step forward and developing another advanced cell switching algorithm, we concentrate on removing the barrier between the state-of-the-art advanced techniques and their real-life implementations, because the majority of the existing algorithms in the literature are impractical due to the traffic load estimation problem.} 

Research on cell switching has extensively focused on more efficient BS deactivation strategies to reduce network energy consumption. For instance, the study in~\cite{10304250} examined a vHetNet model with a HAPS functioning as a super macro base station (SMBS), referred to as HIBS by International Telecommunication Union (ITU)~\cite{itu_vision_november}, alongside a macro BS (MBS) and several SBSs. It aimed to optimize sleep mode management of SBSs and utilize the capabilities of HAPS-SMBS to lower energy use while maintaining user quality-of-service (QoS).
Similarly, research in~\cite{HETS2023P} addressed cell switching challenges with HAPS-SMBS, specifically targeting traffic offloading from deactivated BSs using a sorting algorithm that prioritizes BSs with lower traffic loads for deactivation.
A tiered sleep mode system that adjusts sleep depth according to device activity was proposed in~\cite{ELAA2022JR}, featuring decentralized control for scalability and efficiency.
Additionally, the study in~\cite{EOMK2017JR} considered the control data separated architecture (CDSA) and implemented a genetic algorithm to optimize energy savings in HetNets by managing user associations and BS deactivation through deterministic algorithms.
Furthermore, the work in~\cite{Metin_VFA_CellSwitch} introduced a value function approximation (VFA)-based reinforcement learning (RL) algorithm for cell switching in ultra-dense networks, demonstrating scalable energy savings while maintaining QoS.

However, the challenge of traffic load estimation for sleeping BSs remains unaddressed, limiting the practical application of these studies and their contribution achieving sustainability goals of 6G networks as outlined in the 6G framework of ITU~\cite{itu_vision_june_23}.
Addressing this gap, our work is a pioneering effort to investigate and evaluate the impacts of traffic load estimation errors on the cell switching performance and explores various spatial interpolation methods for estimating the traffic load of sleeping SBSs, which is crucial for the viability and applicability of cell switching strategies. 
We conduct a comparative analysis of estimated and actual traffic loads and assess the implications of any discrepancies on network power consumption within a vHetNet~\cite{10304250}. 
Our study significantly contributes by:

\begin{itemize}
\item Introducing and exploring the problem of load estimation for sleeping SBSs for effective cell switching and traffic offloading in a vHetNet.
\item Implementing three established spatial interpolation algorithms---random neighboring selection, distance-based selection, and $k$-means clustering---for traffic load estimation. 
\item Developing a mathematical framework to understand the behaviors of implemented spatial interpolation techniques.
\item Utilizing a well-established telecommunication dataset based on
real-world call detail record (CDR) information in the city of Milan~\cite{DVN/EGZHFV_2015} to enhance the realism and reliability of our approach.
\item Evaluating the impact of traffic load estimation errors on network power consumption and decision changes in cell switching strategies.
\end{itemize}
\section{System Model}\label{sec:model}
\subsection{Network Model}
Our study explores a vHetNet depicted in Fig. \ref{fig-1}, comprising a macro cell with a single MBS and $n \in \mathbb{N}$ SBSs. Additionally, a HAPS-SMBS is integrated into the network, with the potential to serve multiple macro cells.
The primary function of SBSs is to deliver data services and address user-specific requirements, while MBS and HAPS-SMBS ensure consistent network coverage and manage control signals.
A key role of the HAPS-SMBS is to efficiently manage traffic offloading from SBSs during low-traffic periods, utilizing its extensive line-of-sight~(LoS) and large capacity~\cite{9380673}. 
This capability not only enhances network flexibility and optimizes capacity utilization but also offers a cost-effective alternative to deploying multiple terrestrial BSs, particularly in fluctuating and high-demand scenarios.

\begin{figure}[t]
\centering
\includegraphics[width = .6\columnwidth,trim={5cm 0 8.2cm 0},clip]{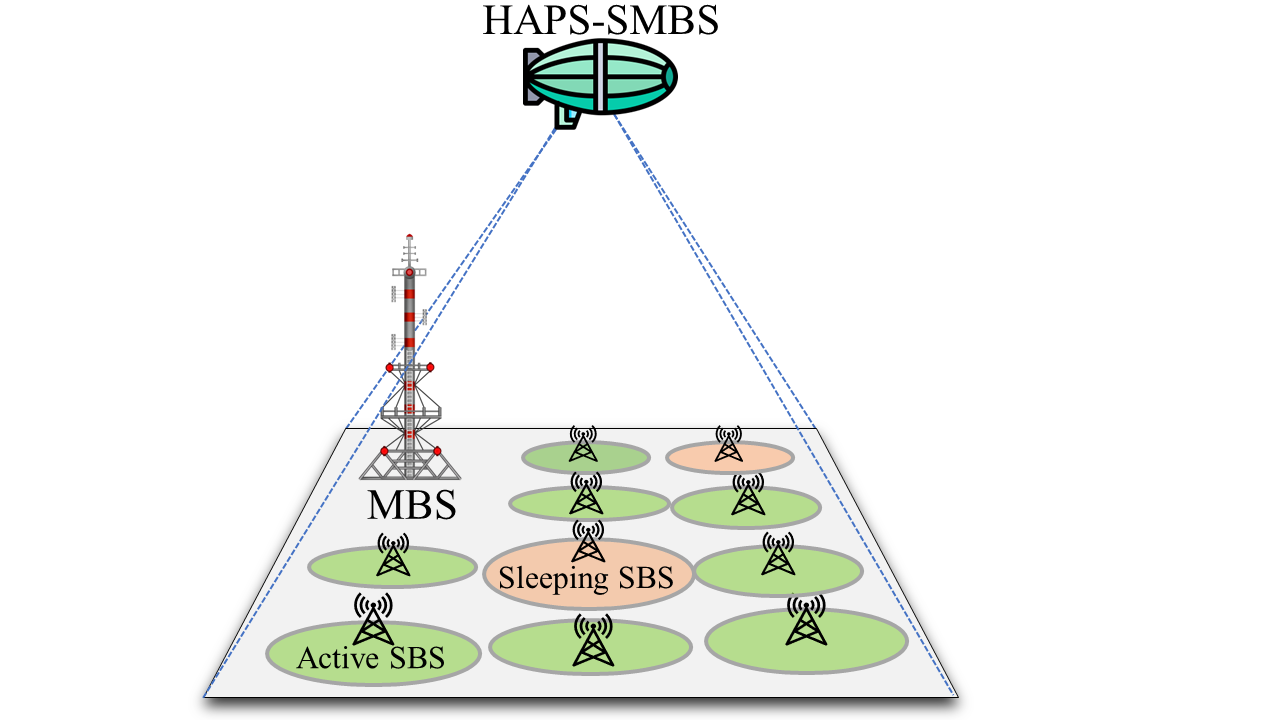}
\caption{A vHetNet model with an MBS, multiple SBSs, and a HAPS-SMBS.}
\label{fig-1}
\vspace{-1em}
\end{figure}
\subsection{Network Power Consumption}
The power consumption of each BS in the network is calculated based on the energy-aware radio and network technologies (EARTH) power consumption model~\cite{6056691}. 
For the $i$-th BS, the power consumption at any given time, denoted as $P_i$, is expressed as \cite{7925662}
\begin{equation}
P_i  = \left\{ {\begin{array}{*{20}c}
   {\begin{array}{*{20}c}
   {P_{\text{o},i}  + \eta _i \lambda _i P_{\text{t},i} }, & {0 < \lambda _i  < 1,}  
\end{array}}  \\
   {\begin{array}{*{20}c}
   {P_{\text{s},i} },  \; \;   \; \; \; \; \; \; \;\; \; \; \; \; \; \; \; \; \; \; \; \; \; & {\lambda _i  = 0},  \\
\end{array}}  \\
\end{array}} \right.
\label{eq1}
\end{equation}
where $P_{\text{o},i}$ represents the operational circuit power consumption, $\eta_i$ is the power amplifier efficiency, $\lambda_i$ is the load factor (i.e., the normalized traffic load), $P_{\text{t},i}$ is the transmit power, and $P_{\text{s},i}$ is the power consumption in sleep mode. 
The total instantaneous power consumption of vHetNet, denoted as $P_\text{T}$, is given by
\vspace{-1mm}
\begin{equation}
P_\text{T} = P_\text{H}  + P_\text{M}  + \sum\limits_{j = 1}^s {P_j }, 
\label{eq2}
\end{equation}
where $P_\text{H}$ and $P_\text{M}$ denote the power consumption of HAPS-SMBS and MBS at any given moment, respectively, which are calculated based on the $(0 < \lambda _i  < 1)$ case in \eqref{eq1} as HAPS-SMBS and MBS are always active in our modeling. 
Meanwhile, $P_j$ represents the power consumption of SBS $j$ and $s$ signifies the total number of SBSs within the network.
\vspace{-1mm}
\subsection{Data Set and Data Processing}
To assess power consumption as defined in~\eqref{eq1}, we require the load factor $\lambda_i$ for each BS. 
Thus, we employ a real call detail record (CDR) data set from Telecom Italia~\cite{DVN/EGZHFV_2015} that captures user activity in Milan, partitioned into 10,000 grids of $235\times 235$ meters. This activity includes calls, texts, and internet usage recorded every 10 minutes over November and December 2013. We consolidate these activities into a single measure of traffic load per grid. Each SBS is then assigned a normalized traffic load from a randomly selected grid to simulate its corresponding cell activity.
% \vspace{-1mm}
\section{Problem Formulation}\label{sec:problem}
\subsection{Optimization of the Total Power Consumption}
Our objective is to minimize the total power consumption of the vHetNet, $P_\text{T}$. We achieve this by identifying the most efficient power-saving state. The state vector ${\rm \Delta = \{
\delta _1 ,\delta _2 ,...,\delta _s \} } $ represents the status (ON or OFF) of each SBS at time $t$, where ${\delta _j  \in \left\{ {0,1} \right\}}$ indicates the state of the $j$-th SBS (0 for OFF, 1 for ON). It is assumed that MBS and HAPS-SMBS are always ON ($\delta _\text{M} = \delta _\text{H} = 1$). 
If SBS $j$ is switched off, i.e.,  when ${\delta _j}$ changes from 1 to 0 at time ${t}$, its traffic load becomes zero (i.e., $\lambda _{j,t}  = 0$) and either MBS or HAPS-SMBS can allocate the load, such that
\begin{equation}
\lambda _{k,t}  = \lambda _{k,(t - 1)}  + \phi _{j,k} \lambda _{j,(t - 1)}, \quad k=\{\text{M},~ \text{H}\}, 
\label{eq3}
\end{equation}
where $ {\lambda}_{k,t} $ is the load factor of MBS (i.e., $k=$ M) or HAPS-SMBS (i.e., $k=$ H), while $ {\lambda}_{j,t} $ is the load factor of the $j$-th ${\rm SBS}$. 
Additionally, ${\phi_{j,k}}$ is the relative capacity between $j$-th ${\rm SBS}$ and the MBS or HAPS-SMBS, such that ${\phi _{j,k}  = {{C_j } \mathord{\left/
 {\vphantom {{C_j } {C_k }}} \right.
 \kern-\nulldelimiterspace} {C_k }}}$. 
 On the other hand, when ${\delta _j }$ changes from 0 to 1, the MBS or HAPS-SMBS offloads some traffic to SBS $j$, such that ${\lambda _{j,t}  = {{\tau _j } \mathord{\left/
 {\vphantom {{\tau _j } {C_j }}} \right.
 \kern-\nulldelimiterspace} {C_j }}}$ and
\begin{equation}
 \lambda _{k,t}  = \lambda _{k,(t - 1)}  - \phi _{j,k} \lambda _{j,t},  \quad k=\{\text{M},~ \text{H}\}, 
 \label{eq5}
\end{equation}
where ${\tau _j }$ refers to the traffic load\footnote{Traffic load, $\tau$, represents the load of a BS, i.e., the amount of bandwidth occupied, while load factor, $\lambda$, denotes the normalized traffic load, which is computed by dividing the traffic load, $\tau$, to the capacity of the BS, $C$.
Since one can be computed from another, they can be used interchangeably throughout the paper.} of $j$-th ${\rm SBS}$. 
Accordingly, the optimization problem can be formulated as
 \begin{IEEEeqnarray*}{lcl}\label{eq:P1}
    &\underset{\mathbf{\Delta}}{\text{minimize}}\,\, & ~P_\text{T}(\Delta ) \,  \IEEEyesnumber \IEEEyessubnumber* \label{eq:P1_Obj}\\
    &\text{s.t.} & {\lambda _\text{M}}{\le 1,} \label{eq:P1_const1}\\
    && {\lambda _\text{H}}{\le 1,} \label{eq:P1_const2}\\
    && {{\delta _j \in \{ 0,1\}, }} \quad {{\; j = 1,...,s}}. \label{eq:P1_const3}
\end{IEEEeqnarray*}

Substituting \eqref{eq1} into \eqref{eq2}, and rearranging the equation as a closed-form expression with the inclusion of $\Delta$, we get
\begin{equation}
\begin{aligned}
P_\text{T} (\Delta )
 &={(P_{\text{o,H}}  + \eta _\text{H} \lambda _\text{H} } P_{\text{t,H}} )+{(P_{\text{o,M}}  + \eta _\text{M} \lambda _\text{M} } P_{\text{t,M}})\\
 &\hspace{-2em}\;\;\;\;\;\;\;+\Bigg(\sum\limits_{j = 1}^s {(P_{\text{o},j}  + \eta _j \lambda _j  P_{\text{t},j} )\delta _j +P_{\text{s},j}(1-\delta _j)\Bigg)}.
\end{aligned}
\label{eq7}
\end{equation}

In the constraints \eqref{eq:P1_const1} and \eqref{eq:P1_const2}, $\lambda_\text{M}=[0,1]$ and $\lambda_\text{H}=[0,1]$ are defined as positive real numbers, $\mathbb{R}^+$, ensuring that the operational capacities of both the MBS and HAPS-SMBS are never exceeded, thereby upholding the QoS requirements. 
Importantly, the optimal state vector, ${\Delta _{\text{opt}}}$, minimizes ${P_\text{T}(\Delta )}$, which is a function of the load factors, ${\lambda _j}$, of SBSs. 
It should be noted that the optimization model in \eqref{eq:P1} presumes complete knowledge of all BSs' traffic loads, including those in sleep mode. This presumption, common in the literature~\cite{10304250,HETS2023P,ELAA2022JR,EOMK2017JR,Metin_VFA_CellSwitch}, poses a pivotal question: How can we accurately determine the traffic load for a sleeping SBS to decide its next state? The lack of precise knowledge about ${\lambda _j}$ for these SBSs necessitates a reliable estimation method. Addressing this estimation challenge is essential for refining our optimization strategy, which we will explore in the subsequent section.
\subsection{Power Consumption with the Erroneous Load Estimations} 
Due to the unavailability of real-time traffic load data for the sleeping SBSs, we consider an estimated load, ${\hat \lambda}$, introducing a layer of uncertainty in the optimization process. Consequently, we redefine the network's total power consumption given in~\eqref{eq7} with the estimated load factors, $\hat \lambda$, as
\begin{equation}
\begin{aligned}
P_\text{est} (\Delta )
 &={(P_{\text{o,H}}  + \eta _\text{H} \hat \lambda _\text{H} } P_{\text{t,H}} )+{(P_{\text{o,M}}  + \eta _\text{M} \hat \lambda _\text{M} } P_{\text{t,M}})\\
 &\hspace{-2em}\;\;\;\;\;\;\;+\Bigg(\sum\limits_{j = 1}^s {(P_{\text{o},j}  + \eta _j \hat \lambda _j  P_{\text{t},j} )\delta _j +P_{\text{s},j}(1-\delta _j)\Bigg)}.
\end{aligned}
\label{eq8}
\end{equation}

Depending on the accuracy of our estimation, this could be either perfect or imperfect estimation. In the case of perfect estimation, where ${\lambda _j}={\hat \lambda _j},~\forall j$, the estimated power consumption, $P_\text{est}$, accurately captures the ground truth, $P_\text{T}$. 
Achieving perfect estimation is the goal of our work and should be the goal of any other work employing a cell switching scheme. 
% However, this need for estimation has not already been introduced in the literature. To the best of our knowledge, this is the first time we are introducing and addressing this gap between perfect knowledge and reality assumption. 

In the case of imperfect estimation, where ${\lambda _j} \neq {\hat \lambda _j},~\forall j$, there would be a probability of error, ${p_\text{err}}$, and we can define an expected value for the total power consumption of the network as
\begin{equation}  
E[P] = P_\text{est}.p_\text{err}  + P_\text{T}.(1 - p_\text{err}).
\label{eq9}
\end{equation}
% Considering the need for accurate traffic load estimation, the following theorem investigates two scenarios that illustrate the potential consequences of estimation inaccuracies on the network's power consumption optimization.
\begin{theorem}
Errors in estimating the traffic load of SBSs in a vHetNet can lead to changes in the optimal state vector, thereby affecting the total power consumption of the network.
\end{theorem}
\begin{proof}
We analyze the impact of traffic load estimation errors on network power consumption through two primary scenarios related to the operational dynamics of SBSs:
\begin{itemize} 
\item \underline{Overestimation Scenario}: If the estimated load ${\hat \lambda _j}$ at time $t_1$ for a sleeping SBS is more than the actual load and a set threshold $\lambda _\text{th}$, it might mistakenly trigger the SBS's transition to an active state, leading to unnecessary power consumption. The probability of this erroneous transition, $p_\text{err} ^{\text{off} \to \text{on}}$, is defined as
\begin{equation}
p_\text{err} ^{\text{off} \to \text{on}}  = \text{Pr}\{ \hat \lambda _j  > \lambda _\text{th}\mid \lambda _j  \le \lambda _\text{th} \}.
\label{eq10}
\end{equation}
The expected value of the associated power consumption error, $P_\text{err}$, due to this transition is calculated as
\begin{equation} 
\begin{array}{l}
E[P_\text{err}]  = \big|(\eta _\text{H} \phi _{j,\text{H}} \lambda _j P_\text{t,H}  + P_{\text{s},j} ) - (P_{\text{o},j} \\ + \eta _j \hat \lambda _j P_{\text{t},j} )\big|\times p_\text{err} ^{\text{off} \to \text{on}}. 
%\.pr\{ \lambda _{i,t_1 }  \le \lambda _\text{th} \} .pr\{ \lambda _{i,t_0}  \le \lambda _\text{th} \} 
\end{array}
\label{eq11}
\end{equation}
\item \underline{Underestimation Scenario}: Conversely, if the actual load of an SBS at time $t_1$ is higher than the estimated and it mistakenly remains in sleep mode, this can lead to insufficient traffic management. 
The probability of such underestimation, $p_\text{err} ^{\text{on} \to \text{off}}$, is defined as
\begin{equation}
p_\text{err} ^{\text{on} \to \text{off}}  = \text{Pr}\{\hat \lambda _j  < \lambda _\text{th}\mid \lambda _j  \ge \lambda _\text{th} \}.
\label{eq12}
\end{equation}
The expected value of the error in the total power consumption of the network, $P_\text{err}$, for this underestimation scenario is given by
\begin{equation} 
\begin{array}{l}
E[P_\text{err}]  = \big|(P_{\text{o},j}  + \eta _j \lambda _j P_{\text{t},j} )-(\eta _\text{H} \phi _{j,\text{H}} \hat\lambda _j P_\text{t,H} \\ + P_{\text{s},j} )\big|\times p_\text{err} ^{\text{on} \to \text{off}}.%\ .pr\{ \lambda _{i,t_1 }  \ge \lambda _\text{th} \} .pr\{ \lambda _{i,t_0}  \le \lambda _\text{th} \} 
\end{array}
\end{equation}
\label{eq13}
\end{itemize}
\vspace{-2mm}

These scenarios illustrate the significant impact that accurate load estimation of sleeping SBSs has on the efficient optimization of network power consumption, affirming the core assertion of this Theorem.
\end{proof}
\vspace{-2mm}
\section{Spatial Interpolation for Traffic Load Estimation}\label{sec:method}
We apply certain spatial interpolation methodologies for estimating the traffic load of sleeping SBSs, a critical factor in optimizing network power consumption. 
\vspace{-1mm}
\subsection{Geographical Distance-Based Traffic Load Estimation}
This method considers the proximity of neighboring cells to estimate the traffic load of a sleeping SBS. It includes two sub-methods based on the presence of a weighting mechanism, which is developed to prioritize the effect of closer cells.
\subsubsection{Distance-Based without Weighting} In this approach, the traffic load of a sleeping SBS is estimated by averaging the traffic loads of its neighboring cells, arranged incrementally based on proximity. All neighboring cells contribute equally to the estimation, regardless of their distance. The estimated traffic load of SBS $j$, $\hat\lambda _j$, is calculated as
\begin{equation}
    \hat\lambda _j  = \frac{1}{N}\sum\limits_{a = 1}^N {\lambda _a,} 
\label{eq14}
\end{equation}
where ${{\lambda_a}}$ represents the traffic load of cell $a$, and $N$ is defined as the number of cells included in the estimation process.
Note that usually $N\gg s$, as $N$ encompasses all the cells available for the estimation process, while $s$ is the number of SBSs within a single vHetNet with a single macro cell.
  
\subsubsection{Distance-Based with Weighting} This method refines the previous approach by assigning different weights to neighboring cells based on their distance from the sleeping SBS. The closer a cell is, the more it influences the estimated traffic load. The weighted traffic load, $\hat\lambda _j$, is calculated as
\begin{equation}
\hat \lambda _j  = {\sum\limits_{a = 1}^N {\lambda _a  \times w_{j,a} } }\Big/{{\sum\limits_{a = 1}^N {w_{j,a} } }},
\label{eq15}
\end{equation}
where the weighting factor, ${w_{j,a}}$, is defined as
\begin{equation} 
w_{j,a}  = \frac{{d_{\max } }}{{d_{j,a}^n }},\quad
n \in \mathbb{R}^+, 
\label{eq16}
\end{equation}
where ${d_{\max }}$ is the maximum distance between the sleeping SBS and its neighboring cells included in the estimation, and ${d_{j,a}}$ is the distance between the sleeping SBS ${j}$ and the neighboring SBS ${a}$.
\begin{theorem}
The error in estimating the traffic load of a sleeping SBS decreases as the exponent $n$, representing the power of the distance between the sleeping SBS and its neighboring cells, increases.
\label{theorem2}
\end{theorem}
\begin{proof}
% If we define the estimation error for BS $i$ as 
% \begin{equation}\label{eq:est_err}
% \delta = |\hat\lambda_i - \lambda_i|,
% \end{equation}
Substituting~\eqref{eq16} into \eqref{eq15}, we get 
\begin{equation}\label{eq:simplified}
    \hat \lambda _j  = \frac{{\sum\limits_{a = 1}^N {\lambda _a  \times\frac{{d_{\max } }}{{d_{j,a}^n }} } }}{{\sum\limits_{a = 1}^N {\frac{{d_{\max } }}{{d_{j,a}^n}}}}}= \frac{{\sum\limits_{a=1}^N \lambda_a \frac{{1}}{{d_{j,a}^n}}}}{{\sum\limits_{a=1}^N \frac{{1}}{{d_{j,a}^n}}}}.
\end{equation}

The $n$-dependent term of~\eqref{eq:simplified}, $\sum_{a=1}^N \dfrac{{1}}{{d_{j,a}^n}}$, is increasingly dominated by the terms with the smallest values of $d^n_{j,a}$ as $n$ increases, that is, the weighting factor, $w_{i,j}=\dfrac{{d_{\max}}}{{d_{j,a}^n}}$, heavily favors smaller distances and diminishes the influence of larger distances.
As $n$ goes to infinity, i.e., $\lim_{n \to \infty}\sum_{a=1}^N{\frac{{d_{\max}}}{{d_{j,a}^n}}}$ , the effect of larger distances (i.e., $d_{j,a}>1$) on the estimation approaches to zero, while smaller distances (i.e., $d_{j,a}<1$) have significant effects, such that 
\begin{equation}\label{eq:limit_to_infinity}
    \lim_{n \to \infty}\frac{{1}}{{d_{j,a}^n}} =\begin{cases}
        0, & d_{j,a}>1,\\
        1, & d_{j,a}=1,\\
        \infty, & d_{j,a}<1,
    \end{cases}
\end{equation}
indicating that while the weights of smaller distances increase infinitely, the weights of larger distances diminish, i.e., have no impact on the estimation process.
Furthermore, it can be observed from \eqref{eq:simplified} and \eqref{eq:limit_to_infinity} that the influence on the estimation process becomes increasingly concentrated at even closer distances as $n$ increases; in other words, the range of close distances narrows with growing $n$.
\begin{prop}\label{prop:spatial}
The BSs in closer proximity have more correlated traffic patterns.
\end{prop}
\begin{proof}
This kind of statements usually require an empirical proof. 
In this regard, the authors in~\cite{spatial-proof} found that the traffic patterns of the closer BSs in terms of distance are correlated to each other. 
Similarly, the study in~\cite{spatial-2} demonstrated a spatial correlation in BS traffic patterns.
Additionally, the work in~\cite{spatial-milan-2}, which specifically analyzed the same dataset as our study, though not explicitly stated, suggests from their analyses that the traffic of closer cells exhibits greater similarity.
\end{proof}

Considering the results derived from \eqref{eq:simplified} and Proposition \ref{prop:spatial}, it becomes evident that as the estimation process is increasingly influenced by cells closer to the cell-in-question (as indicated by the growing value of $n$), the estimation error, $\epsilon$, given by $\epsilon = |\hat\lambda_j - \lambda_j|=\left|\dfrac{{\sum\limits_{a = 1}^N {\lambda _a  \times\frac{{d_{\max } }}{{d_{j,a}^n }} } }}{{\sum\limits_{a = 1}^N {\frac{{d_{\max}}}{{d_{j,a}^n }} } }} - \lambda_j\right|$ tends to be lower, since the estimation has more correlated components.
\end{proof}

The parameter $n$ plays a crucial role in the estimation process, such that the increasing values of $n$ make the estimation more immune to the growing number of SBSs, $N$, in the network.

\begin{corollary}\label{cor:moreN}
 For a fixed value of $n$, the growing values of $N$ make the traffic load estimation worse, such that the estimation error, $\epsilon$ increases with $N$.   
\end{corollary}
\begin{proof}
As the number of SBSs, $N$, increases, more BSs that are farther from the cell-in-question contribute to the cell load estimation process, which, in turn, boosts the cell load estimation error, $\epsilon$.
\end{proof}

\begin{lemma}
The deviation in the estimation error with increasing values of $N$ reduces with the growing values of $n$.
\label{lemma1}
\end{lemma}
\begin{proof}

Corollary~\ref{cor:moreN} states that the increasing values of $N$ tend to increase the traffic load estimation error, $\epsilon$, such that $\dfrac{\text{d}\epsilon (N)}{\text{d}N}>0$.
On the other hand, Theorem~\ref{theorem2} indicates that the estimation error, $\epsilon$, reduces with growing values of $n$, i.e., $\dfrac{\text{d}\epsilon(n)}{\text{d}n}<0$.
When the effects of increasing values of $N$ and $n$ are combined, the effect of increasing $N$ is mitigated with the effect of increasing $n$, such that $\dfrac{\text{d}\epsilon (N)}{\text{d}N}\cdot \dfrac{\text{d}\epsilon(n)}{\text{d}n}<0$, and thereby if both $N$ and $n$ grow together, the deviation in the estimation error reduces.
\end{proof}

\subsection{Random Cell Selection Traffic Load Estimation}
This approach utilizes a random selection of surrounding cells for traffic load estimation. It comprises two variations: random selection without weighting and random selection with weighting.
\subsubsection{Random Selection without Weighting}
The traffic load of a sleeping SBS is estimated based on the average traffic load of randomly selected surrounding SBSs. The estimation is calculated using a formula similar to \eqref{eq14}, where the selection of neighboring cells, $a$, is random.
\subsubsection{Random Selection with Weighting}
This variation applies a weighting mechanism to the randomly selected surrounding cells, the same as \eqref{eq15}, with the selection of $a$ being random. The weighting factors vary based on the distance to the sleeping SBS, enhancing the accuracy of the traffic load estimation. 
\subsection{Clustering-Based Traffic Load Estimation}
This approach involves clustering SBSs based on their traffic patterns and estimating the traffic load of a sleeping SBS using the average load of active SBSs within the same cluster. 
We employ the $k$-means algorithm, an unsupervised machine learning technique, for clustering the SBSs. 
The number of clusters, a crucial hyper-parameter in the $k$-means algorithm, is determined using the elbow method \cite{article}, which
assesses different cluster numbers by calculating the sum of squared errors (SSE) between data points and centroids for each potential cluster count. The SSE is given by \cite{article}
\begin{equation} 
SSE = \sum\limits_{g = 1}^G {\sum\limits_{x_m\in \kappa_g} {(x_m - \varkappa_g )} } ^2,
\label{eq18}
\end{equation}
where $G$ represents the optimal number of clusters, $\varkappa_g$ the centroid of each cluster $\kappa_g$, and $x_m$ each sample in $\kappa_g$. 
The optimal cluster number is identified at the point where the SSE curve forms an "elbow" before flattening.
\subsubsection{Multi-level Clustering-Based Traffic Load Estimation}
The multi-level clustering (MLC) approach involves repeated clustering of SBSs based on their traffic patterns to estimate the traffic load of offloaded SBSs. 
This method employs the elbow method to determine the optimum number of clusters, followed by the application of the $k$-means algorithm for clustering. The traffic load of a sleeping SBS is then estimated based on the average load of active SBSs in the same cluster. This iterative approach, as outlined in Algorithm 1, ensures progressively refined clustering with each layer, leading to more precise traffic load estimations for sleeping SBSs.
\vspace{-2mm}
\begin{algorithm}
\caption{Multi-Level Clustering (MLC) using $k$-means}
\DontPrintSemicolon  % Removes semicolons
\SetAlgoLined  % Enables vertical lines
\small
\KwData{Traffic loads of SBSs $\lambda_{a}$, maximum number of layers $L$}
\KwResult{Clustered SBSs with estimated traffic loads}

\SetKwFunction{FMain}{MLC\_k\_means}
\SetKwProg{Fn}{Procedure}{:}{}
\Fn{\FMain{$\lambda$, $L$}}{
    Determine the optimal number of clusters $G$ using the elbow method\;
    Initialize layer count $l = 1$\;
    \While{$l \leq L$}{
        Perform $k$-means clustering on $\lambda$ to form $G$ clusters\;
        \For{cluster $\kappa_g$}{
            Calculate the mean traffic load $\mu_\text{m}$\;
            \For{sleeping SBS in $\kappa_g$}{
                Estimate the traffic load as $\mu_\text{m}$\;
            }
        }
        Update $\lambda$ with estimated ones for sleeping SBSs\;
        Increment the layer count $l$ by 1\;
    }
    \Return The final clusters with estimated traffic loads\;
}
\label{alg1}
\end{algorithm}
\vspace{-0.5cm}
\section{Performance Evaluation}\label{sec:performance}
This section evaluates the effectiveness of our proposed traffic load estimation methods for sleeping SBSs using the Milan dataset detailed in Section II-C. Simulation parameters are provided in Table~\ref{table:nonl}. We average monthly traffic per time slot for each SBS and apply our algorithms across random SBS selections for each iteration.
\subsection{Performance Metric}
Three different metrics are utilized to evaluate the performance: estimation error, network power consumption, and decision change rate. \textit{Estimation error} is defined as the ratio of the difference between the actual and estimated traffic load to the actual traffic load. \textit{Network power consumption} values are derived from \eqref{eq7} and \eqref{eq8} for scenarios with true and estimated values, respectively. \textit{Decision change} represents the discrepancy between the actual and estimated state vectors. This metric reflects the impact on users' QoS, indicating that any change in the state vector due to cell switching results in a reduction of QoS, as user association under normal circumstances, where all SBSs are active, is optimized.
\begin{centering}
\begin{table}[t!] 
\caption{SIMULATION PARAMETERS}
\centering % used for centering table 
\begin{tabular}{|c|c|} % centered columns 
  \hline % inserts single horizontal line
Parameters&Value
\\ 
\hline
Number of SBSs&5000
\\
Number of time slots&144
\\
Time slot duration&10 m
\\
Number of days&60
\\
Number of iteration&300
\\
Optimal $G$ using elbow method&3
\\
\hline
    \end{tabular} \label{table:nonl} % is used to refer this table in the text 
    \vspace{-1.5em}
    \end{table}
    \end{centering}
\subsection{Simulation Results}
Figure \ref{fig-2} presents the relationship between the traffic load estimation error and the number of neighboring SBSs for the distance-based with weighting estimation method. 
The analysis, conducted for different powers of distance, \(n\), in \eqref{eq16} specifically $n=$ 1, 3, 5, and 10, provides key insights that are consistent with our theoretical findings:
\begin{itemize}
\item As Corollary \ref{cor:moreN} suggests, an increase in the number of neighboring SBSs (\(N\)) elevates estimation errors.
\item Consistent with Theorem \ref{theorem2}, an increase in \(n\) significantly reduces estimation errors. The errors decrease from 45\% at \(n = 1\) to around 15\% at \(n = 5\), and become minimal at \(n = 10\), highlighting the benefits of prioritizing proximity in weight calculations.
\item This pattern aligns with Proposition \ref{prop:spatial}, affirming that traffic loads from closer BSs are more correlated and thus provide more reliable estimations.
\end{itemize}

\begin{figure}[t]
\centerline{\includegraphics[width = 6.7cm ]{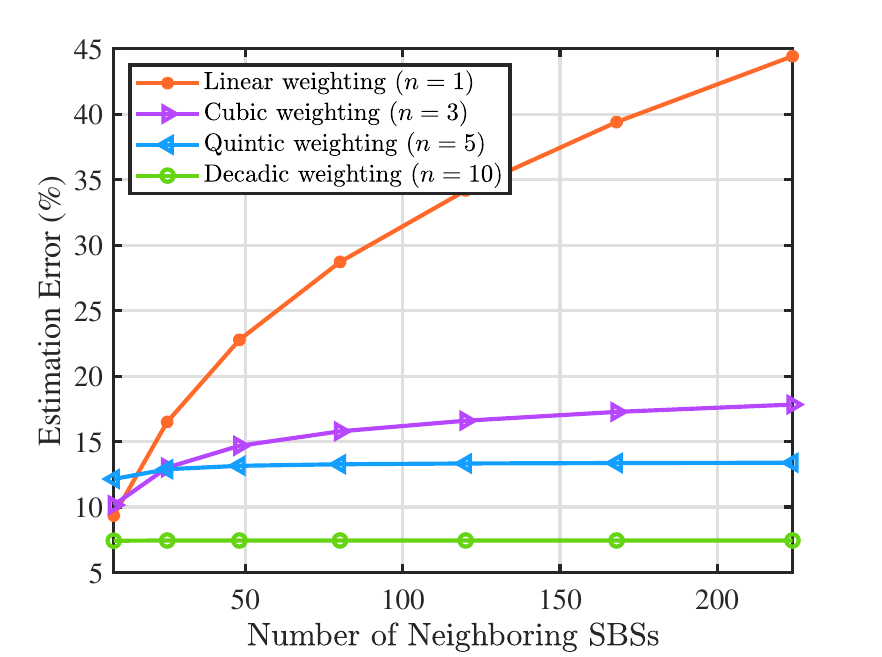}}
\vspace{-2mm}
\caption{Estimation error for distance-based cell selection with weighting.}
\label{fig-2}
 \vspace{-1em}
\end{figure}
\begin{figure}[h]
\centerline{\includegraphics[width = 6.7cm ]{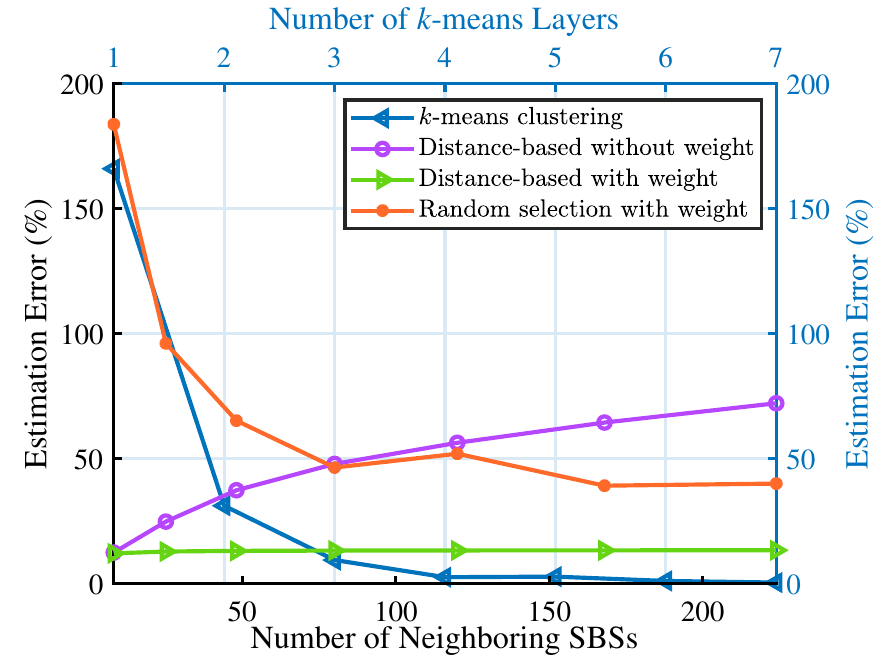}}
\vspace{-3mm}
\caption{The estimation error for different methods. Two $x$- and two $y$-axes are considered: the blue one is for the clustering-based while the black one is for the rest of the approaches.}
\label{fig-3}
\vspace{-1em}
\end{figure}
Figure \ref{fig-3} evaluates the estimation error across different methods. With an increase in the number of layers in $k$-means clustering, the homogeneity within clusters improves, leading to a substantial reduction in estimation errors—approaching near 0\% with seven layers. The distance-based method without weighting shows a rising trend in errors as the number of neighboring cells increases, emphasizing the relevance of incorporating distance in estimation, as suggested by Theorem \ref{theorem2} and Proposition \ref{prop:spatial} which highlight the importance of closer proximities in reducing errors. In contrast, the performance of the distance-based method with weighting, particularly at $n=10$, shows minimal errors, as extensively analyzed in Fig. \ref{fig-2}.
The random selection method with weighting shows decreased errors due to prioritizing closer neighbors, similar to the effects described in Theorem \ref{theorem2}. The random selection method without weighting is excluded due to its consistently high and unstable errors, highlighting the importance of distance consideration in neighbor selection as reinforced by Lemma~\ref{lemma1}.

Figure~\ref{fig-4} depicts the network power consumption as a function of the number of SBSs, $s$, under both actual and estimated traffic loads, the latter obtained via $k$-means clustering with different numbers of layers, $L$. 
The results demonstrate that network power consumption increases with $s$. Furthermore, increasing $L$ leads to a closer alignment of estimated power consumption with the actual values, indicating improved accuracy in traffic load estimation with more layers in the MLC process.
\begin{figure}[t]
\centerline{\includegraphics[width = 6.7cm ]{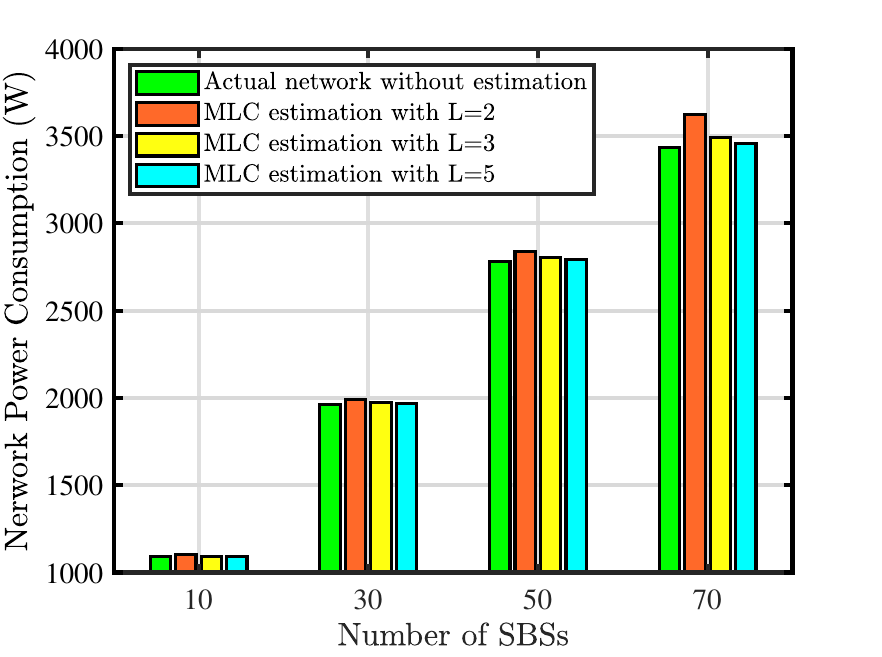}}
\vspace{-3mm}
\caption{Network power consumption for different number of layers in MLC.}
\label{fig-4}
\vspace{-0.5em}
\end{figure}

\begin{figure}[h]
\centerline{\includegraphics[width = 6.7cm ]{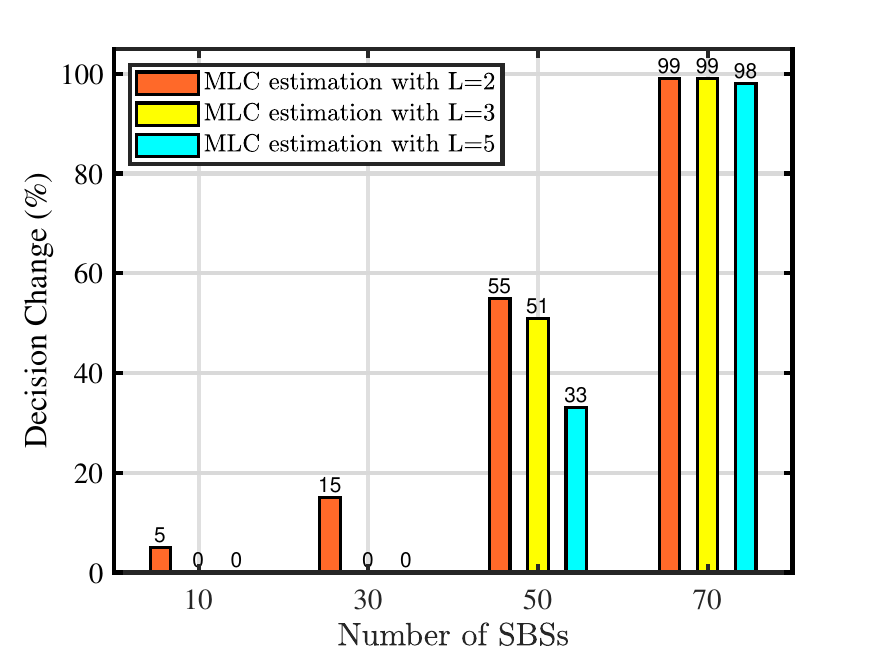}}
\vspace{-3mm}
\caption{Decision change for different number of layers in MLC.}
\label{fig-5}
\vspace{-1em}
\end{figure}
Figure \ref{fig-5} shows the decision change rate between actual and estimated state vectors using $k$-means clustering with various layers. The data shows that the state vector variation increases with the number of SBSs. Adding more layers in the MLC process reduces this discrepancy when the number of SBSs is small. However, in densely populated vHetNets, such as with 70 SBSs, the decision change rate remains high at 98\%-99\%, regardless of the number of layers. While these decision changes have minimal impact on network power consumption due to the efficient offloading algorithm in~\cite{10304250}, they significantly affect user QoS.

It is important to note that while the results are based on the Milan dataset, we believe these findings offer generic insights applicable across various network scenarios, despite potential dataset-specific limitations. This emphasizes the broad relevance of our methodologies.
\section{Conclusion}\label{sec:conclusion}
This study addresses the challenge of estimating traffic loads for SBSs in sleep mode within vHetNets, a significant barrier to optimizing power consumption feasibly and practically. 
After developing mathematical frameworks to explain the behaviors of the employed spatial interpolation methods, we explored their potential to bridge the data gap for cell switching strategies and enhance power consumption management.
Validated with the Milan dataset, these methods have proven effective in accurately estimating traffic loads, enabling the implementation of more efficient energy-saving strategies.
Our approach underscores the potential for significant improvements in network sustainability.

\ifCLASSOPTIONcaptionsoff
  \newpage
\fi
\bibliographystyle{IEEEtran}
\small
\bibliography{IEEEabrv,Bibliography}

\vfill
\end{document}